%% file: frame_coded_locomotion.tex
\documentclass[journal,10pt]{IEEEtran}

\usepackage{amsmath,amssymb,amsfonts,mathtools,bm}
\usepackage{amsthm}
\usepackage{graphicx}
\usepackage{booktabs,tabularx,array,multirow}
\usepackage{cite}
\usepackage{url}
\usepackage{microtype}
\usepackage{amsmath,amssymb}
\usepackage{graphicx}
\usepackage{tikz}
\usetikzlibrary{
  arrows.meta,
  calc,
  decorations.pathmorphing
}
\usepackage[caption=false,font=footnotesize]{subfig}
\usepackage{balance}
\usepackage{enumitem}
\usepackage{xcolor}

\newtheorem{theorem}{Theorem}
\newtheorem{proposition}[theorem]{Proposition}
\newtheorem{lemma}[theorem]{Lemma}
\newtheorem{corollary}[theorem]{Corollary}
\newtheorem{definition}[theorem]{Definition}

\newcommand{\R}{\mathbb{R}}
\newcommand{\E}{\mathbb{E}}
\newcommand{\Pp}{\mathbb{P}}
\newcommand{\tr}{\operatorname{tr}}
\newcommand{\rank}{\operatorname{rank}}
\newcommand{\diag}{\operatorname{diag}}

\newcommand{\argmax}{\operatorname*{arg\,max}}

\newcommand{\Cov}{\operatorname{Cov}}
\newcommand{\Var}{\operatorname{Var}}

\newcommand{\nullsp}{\operatorname{null}}

\newcommand{\Dkl}{D_{\mathrm{B}}}
\newcommand{\loc}{\mathrm{loc}}

\newcommand{\fail}{\mathrm{fail}}
\newcommand{\T}{\mathsf{T}}

\newcommand{\Dset}{\mathcal{D}}
\newcommand{\cN}{\mathcal{N}}
\newcommand{\bmat}[1]{\begin{bmatrix}#1\end{bmatrix}}

\begin{document}

\title{Frame-Coded Legged Locomotion over Noisy Terrain}

\author{Lav~R.~Varshney%
\thanks{L. R. Varshney is with the Department of Electrical and Computer Engineering and the AI Innovation Institute, Stony Brook University, Stony Brook, NY 11794 USA.}
\thanks{ChatGPT 5.6-Sol was used to support writing, drawing figures, experiment implementation, and mathematical development.}}

\maketitle

\begin{abstract}
Open-loop multilegged locomotion over rough terrain has been interpreted as matter transport over a noisy channel: leg--ground interactions are discrete basic active contacts, terrain deletes or perturbs those contacts, and spatial redundancy concentrates the resulting thrust and arrival time.  That construction is repetition-like, however, because every module carries essentially the same scalar locomotion task.  It consequently provides neither a positive task rate nor a decoder that changes with the surviving contact set.  This paper formulates locomotion instead as a quantized finite-frame expansion with erasures.  A $d$-dimensional body-level command is mapped into $N>d$ heterogeneous local contact commands.  Rough terrain erases or corrupts frame coefficients, while a contact-gated compliant morphology physically realizes the weighted active-subframe decoder.  For a linear-Gaussian model, mechanical equilibrium is exactly the posterior mean, tangent stiffness is posterior precision, and mechanical compliance is posterior covariance.  Equal-norm Parseval frames are shown to be minimax optimal against one missing contact, two-contact robustness is governed by frame coherence, and a harmonic frame gives a directly realizable gait family.  For independently surviving contacts of probability $q$, random Gaussian gait frames admit exact reconstruction at every analog dimension rate $R<q$, with a binomial reliability exponent, whereas recovery of arbitrary commands is impossible for $R>q$.  Residual contact noise yields an asymptotic per-mode amplification $1/(q-R)$ and a vanishing mechanical stiffness margin at the threshold.  Mutual information is the logarithmic gain in stiffness volume, leading to an information--locomotion inequality and an exact incremental-redundancy rule that directs the next gait component toward the softest task-relevant unresolved mode.  The resulting analog frame-coding theorem establishes a finite relative redundancy and converse as part of a fundamental limit theory of legged locomotion.
\end{abstract}
\begin{IEEEkeywords}
Analog coding, compliant mechanisms, finite frames, I--MMSE, legged locomotion, morphological computation, random matrices
\end{IEEEkeywords}

\section{Introduction}
Multilegged matter transport is a core approach to robotic and natural locomotion.  To understand the limits of this problem, Chong et al.\ introduced the basic active contact (\emph{bac}) as a discrete terrain interaction and modeled rough terrain as a source of delayed, shortened, and completely missing contacts \cite{ChongEtAl2023}; serially connected leg pairs share contact disturbances and thereby produce reliable and punctual open-loop transport as spatial redundancy grows.  In their communication theory-inspired mathematical model, however, the modules have an identical instantaneous thrust function and so their coupling behaves simply as a moving-average filter.  Like a repetition code, increasing the number of contacts drives the useful dimension per contact to zero.  The limiting thrust constant that is established is therefore a drift level, rather than a positive task capacity as in information theory and so a broader behavioral repertoire is not considered.

Here we aim to move beyond repetition to consider a genuine redundant representation, building on the finite-frame source--channel construction of Goyal, Kova\v{c}evi\'c, and Kelner \cite{GoyalKovacevicKelner2001}. A walking robot has the same architecture once a global body command is distributed into heterogeneous local leg commands.  This approach from harmonic analysis builds on quantized overcomplete expansions where a vector is expanded into more coefficients than its dimension, the coefficients are quantized and some are erased, and the source is reconstructed from the surviving subframe \cite{GoyalVetterliThao1998}.   Kova\v{c}evi\'c, Dragotti, and Goyal extend to oversampled filter banks and hence to signal streams with memory \cite{KovacevicDragottiGoyal2002}; that extension is appropriate for continuous gait sequences and colored terrain (thinking of rugged terrain as a source of noise), but the finite-frame model is the starting point.

The central analogy driving the paper is therefore:
\begin{itemize}
\item global locomotion command $\longleftrightarrow$ source vector,
\item 
heterogeneous leg commands $\longleftrightarrow$ frame coefficients
\item missing or weak contacts $\longleftrightarrow$ erasures and noise, and
\item compliant-body equilibrium $\longleftrightarrow$ active-subframe synthesis.
\end{itemize}
The last correspondence is the essential physical computation (or embodied intelligence) step to move beyond repetition: the physical body itself enacts the decoding computation.  A loaded elastic element whose extension depends on a body coordinate vector contributes a rank-one stiffness whereas an unloaded foot contributes nothing.  The current force paths thus assemble the normal equations for the surviving contacts, and damped relaxation physically solves them.  In a linear-Gaussian model, this physical equilibrium is not just analogous to minimum mean-square error (MMSE) decoding, it is the posterior mean itself.

The paper makes six main contributions.  First, it defines a positive locomotion dimension rate $R=d/N$ by encoding a $d$-dimensional body-level command into $N$ local contacts, which has operational significance in locomotion.  Second, it gives a physical realization theorem: equilibrium of a contact-gated quadratic morphology equals the active-set weighted-MMSE estimate, with stiffness equal to posterior precision.  Third, it translates finite-frame erasure design into gait design; a slightly related approach due to Habala et al.\ draws on statistical physics \cite{Habala2026}.  Equal-norm Parseval frames make every leg equally dispensable; low coherence protects against pairs of lost contacts; and spatial Fourier modes yield a mechanically interpretable harmonic gait frame.  Fourth, it proves an analog coding theorem.  With independent contact-survival probability $q$, random Gaussian gait frames recover arbitrary commands reliably for $R<q$, while no frame decoder can do so for $R>q$.  Fifth, it links information to mechanics through an information--compliance identity and an I--MMSE relation \cite{GuoShamaiVerdu2005}.  Sixth, it develops an incremental-redundancy extension, including an exact one-step information rule, a task-MMSE rule, a submodular approximation guarantee, and a zero-error variable-length theorem.  Fig.~\ref{fig:schematic} shows a schematic of how everything works.

The analog theorem concerns real locomotion dimensions per attempted scalar contact, showing that arbitrarily reliable transport can coexist with a finite asymptotic redundancy factor, because the useful locomotion task repertoire grows with the number of contacts.  

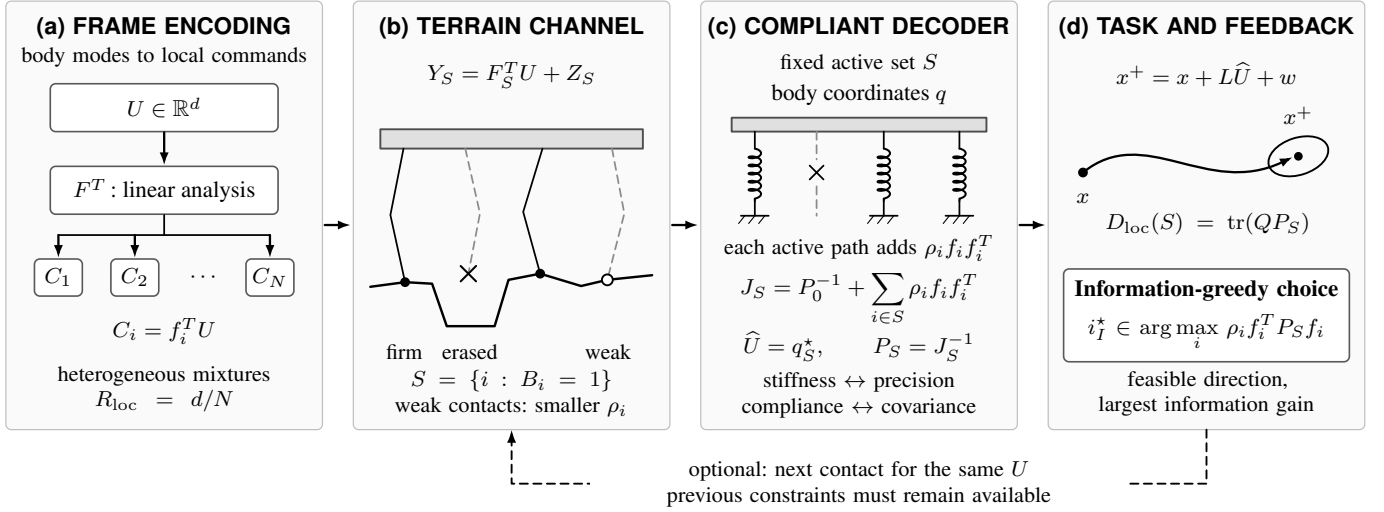
\begin{figure*}[t]
  \centering
  \resizebox{\textwidth}{!}{%
    \input{figures/fig1.tikz}%
  }
  \caption{Frame-coded locomotion over rough terrain.  A body-level command is expanded into heterogeneous local contact commands.  Terrain erases or corrupts a subset, yielding a surviving subframe.  Contact-gated elastic elements physically assemble the active information/stiffness matrix and relaxation computes the weighted-MMSE estimate.  The decoded command determines displacement, pose, and arrival-time performance; feedback can request a new contact direction that is informative for the remaining soft mode.}

  \label{fig:schematic}
\end{figure*}

\section{Frame-Coded Locomotion}
Here we introduce the basic approach.

\subsection{Body-Level Commands and Local Contact Coefficients}
Consider one quasi-static gait epoch.  Let
\begin{equation}
U\in\R^d                                                   \label{eq:command}
\end{equation}
be a desired locomotion command.  Its coordinates may be generalized wrench components, body twist components, support-plane pose variables, backbone modes, or coefficients of a short trajectory segment, consistent with the template-and-anchor and reduced-coordinate traditions in locomotion modeling \cite{FullKoditschek1999,HolmesEtAl2006}.  The high-level gait generator knows $U$, but the physical body has to realize it through distributed local legs (after uncertain terrain interaction).

Let
\begin{equation}
F=\begin{bmatrix} f_1&f_2&\cdots&f_N\end{bmatrix}\in\R^{d\times N},\quad N>d, \label{eq:frame}
\end{equation}
be a finite frame \cite{Christensen2016}.  Contact channel $i$ is assigned the coefficient
\begin{equation}
C_i=f_i^{\T}U.                                             \label{eq:analysis}
\end{equation}
Depending on the mechanism, $C_i$ may specify a tendon displacement, rest length, contact impulse, thrust amplitude, phase perturbation, or local body-wave sample.  The analog dimension rate and redundancy factor are
\begin{equation}
R_{\loc}=\frac{d}{N},\quad \eta=\frac{N}{d}=R_{\loc}^{-1}. \label{eq:rate}
\end{equation}
This rate is positive when $d$ and $N$ grow proportionally.  Scalar repetition corresponds to fixed $d=1$ and $N\to\infty$, hence $R_{\loc}\to0$.

To compare different $N$ under fixed nominal effort, we usually impose the Parseval normalization
\begin{equation}
FF^{\T}=I_d.                                               \label{eq:parseval}
\end{equation}
A \emph{uniform} Parseval frame also satisfies
\begin{equation}
\|f_i\|^2=\frac{d}{N},\quad i=1,\ldots,N.                \label{eq:uniform}
\end{equation}
Then every contact carries the same coefficient energy under an isotropic command ensemble.

Drawing on \cite{GoyalKovacevicKelner2001}, we further explicitly include quantization.  A finite-resolution actuator may receive
\begin{equation}
\widetilde C_i=Q_{\Delta_i}(f_i^{\T}U),                    \label{eq:quant}
\end{equation}
where $\Delta_i$ is the local command resolution.  With subtractive dither, or under the usual high-resolution approximation, the quantization error is zero-mean and contributes variance $\Delta_i^2/12$ \cite{GoyalVetterliThao1998,GrayNeuhoff1998}.  It can therefore be incorporated into the forthcoming contact-noise covariance, but also see \cite{RappDG2019}.

\subsection{Terrain as an Erasure-and-Noise Channel}
Let $B_i\in\{0,1\}$ indicate whether contact $i$ becomes load bearing.  The active set is
\begin{equation}
S=\{i:B_i=1\}.                                             \label{eq:active}
\end{equation}
A missing step is an erasure; a brief, slipping, compliant, or otherwise unreliable contact is represented by a small value.  Conditioned on $S$, the surviving local constraints are
\begin{equation}
Y_S=F_S^{\T}U+Z_S,\quad Z_S\sim\cN(0,R_S),                \label{eq:channel}
\end{equation}
where $F_S$ collects the columns indexed by $S$ and we consider Gaussian noise.  For independent disturbances, as in rugged terrain,
\begin{equation}
R_S^{-1}=\diag(\rho_i:i\in S),    \label{eq:precision}
\end{equation}
with $\rho_i$ being the effective contact level.  Duration, normal impulse, friction margin, actuator resolution, and sensor quality can all enter $\rho_i$.  Erasure is equivalently $\rho_i=0$.  Within each prescribed gait epoch, the contact-availability process is assumed exogenous to the command. Conditional on the active set, contact noise is independent of the command and has the stated covariance. Frame matrices and contact-noise parameters are treated as known.

A Gaussian prior
\begin{equation}
U\sim\cN(\mu,P_0),\qquad P_0\succ0,                       \label{eq:prior}
\end{equation}
will be used for the MMSE and information results in the sequel as an ensemble over commanded gait segments.  The posed inference problem will describe what global body state is realized by the corrupted distributed constraints.

\subsection{Locomotion Loss}
Let $x$ denote a body pose or task state and suppose a local locomotion map is
\begin{equation}
x^+=x+L\widehat U+w,                                      \label{eq:taskmap}
\end{equation}
where $\widehat U$ is the physically reconstructed command and $w$ collects unmodeled dynamics.  For a task metric $W\succeq0$, define
\begin{equation}
Q=L^{\T}WL\succeq0.                                       \label{eq:Q}
\end{equation}
The one-epoch locomotion distortion is
\begin{equation}
d_{\loc}(U,\widehat U)=(U-\widehat U)^{\T}Q(U-\widehat U). \label{eq:locloss}
\end{equation}
Appropriate choices of $L$ and $W$ produce displacement, attitude, generalized-wrench, or first-order arrival-time losses.  Over $T$ epochs,
\begin{equation}
\widehat D_T-D_T=\sum_{t=0}^{T-1}L_t(\widehat U_t-U_t)+\sum_{t=0}^{T-1}w_t. \label{eq:destination}
\end{equation}
Thus command-reconstruction covariance becomes destination and estimated-time-of-arrival covariance through ordinary linear propagation.

\begin{definition}[Frame-coded gait]
A frame-coded gait of dimension $d$ and length $N$ consists of an analysis frame $F\in\R^{d\times N}$, a local command map \eqref{eq:analysis}, a terrain law for $(S,Y_S)$, and a physical or computational synthesis map $g_S:Y_S\mapsto\widehat U$.  It is \emph{exactly recoverable} on $S$ when $g_S(F_S^{\T}u)=u$ for every $u\in\R^d$, and \emph{stably recoverable} when the reconstruction operator has uniformly bounded noise amplification.
\end{definition}

\section{Morphological Active-Subframe Decoding}

This section focuses on the embodied computation of decoding in the presence of noise.

\subsection{Quadratic Morphology}
Let $q\in\R^d$ denote the realized body-level state.  Suppose active contact $i$ contributes the elastic energy
\begin{equation}
V_i(q)=\frac{\rho_i}{2}\left(f_i^{\T}q-y_i\right)^2,        \label{eq:contactenergy}
\end{equation}
and the body has intrinsic energy
\begin{equation}
V_0(q)=\frac12(q-\mu)^{\T}P_0^{-1}(q-\mu).                \label{eq:priorenergy}
\end{equation}
For active set $S$,
\begin{equation}
V_S(q)=V_0(q)+\frac12(F_S^{\T}q-y_S)^{\T}R_S^{-1}(F_S^{\T}q-y_S). \label{eq:totalenergy}
\end{equation}
The active set need not be enumerated explicitly as a computation, the body does it naturally.  A slack tendon or unloaded series-elastic element simply removes its energy term.

\begin{theorem}[Physical posterior-mean decoder]\label{thm:physical}
Under \eqref{eq:channel} and \eqref{eq:prior}, $V_S$ is strictly convex and has the unique equilibrium
\begin{equation}
q_S^\star=J_S^{-1}\left(P_0^{-1}\mu+F_SR_S^{-1}y_S\right), \label{eq:posterior_mean}
\end{equation}
where
\begin{equation}
J_S=P_0^{-1}+F_SR_S^{-1}F_S^{\T}.                          \label{eq:stiffness}
\end{equation}
Moreover,
\begin{equation}
q_S^\star=\E[U\mid Y_S=y_S,S],\qquad
P_S\triangleq\Cov(U\mid Y_S,S)=J_S^{-1}.                 \label{eq:compliance}
\end{equation}
Thus the tangent mechanical stiffness is posterior precision and the small-signal compliance is posterior covariance, up to the units used to nondimensionalize energy.
\end{theorem}
\begin{IEEEproof}
Differentiating \eqref{eq:totalenergy} gives
\begin{equation*}
\nabla V_S(q)=P_0^{-1}(q-\mu)+F_SR_S^{-1}(F_S^{\T}q-y_S).
\end{equation*}
Its Hessian is $J_S\succ0$, so setting the gradient to zero yields \eqref{eq:posterior_mean}.  Completing the square in the Gaussian prior times Gaussian likelihood gives posterior precision $J_S$, posterior covariance $J_S^{-1}$, and posterior mean \eqref{eq:posterior_mean}.
\end{IEEEproof}

\begin{corollary}[Canonical surviving-subframe synthesis]\label{cor:pseudoinverse}
If $P_0^{-1}\to0$, $R_S=\sigma^2I$, and $F_S$ spans $\R^d$, then
\begin{equation}
q_S^\star=(F_SF_S^{\T})^{-1}F_Sy_S=(F_S^{\T})^\dagger y_S. \label{eq:pseudoinverse}
\end{equation}
In the noiseless case, $y_S=F_S^{\T}u$ implies $q_S^\star=u$.
\end{corollary}

\begin{lemma}[Rank-one mechanical realization]\label{lem:rankone}
Let a linearized spring or tendon have extension
\begin{equation*}
\ell_i(q)-\ell_{i,0}=f_i^{\T}q-y_i
\end{equation*}
and stiffness $\rho_i$.  Its tangent stiffness contribution is $\rho_i f_if_i^{\T}$ and its generalized force contribution is $\rho_i f_i y_i$.  Consequently, any positive semidefinite matrix of the form $\sum_i\rho_i f_if_i^{\T}$ can be assembled from such elements; unilateral contact makes the sum active-set dependent.
\end{lemma}
\begin{IEEEproof}
The element energy is \eqref{eq:contactenergy}.  One differentiation gives stiffness $\rho_i f_i(f_i^{\T}q-y_i)$; a second gives the generalized load vector $\rho_i f_if_i^{\T}$.  Energies and Hessians add.
\end{IEEEproof}

The lemma specifies a morphology rather than only an estimator.  Tendon moment arms or attachment geometry encode $f_i$; series elasticity encodes $\rho_i$; and slackness or contact complementarity gates the term.  Elastic networks have already been shown capable of implementing trainable input--output responses through equilibrium physics \cite{SternEtAl2021,AltmanEtAl2024}.  Here the desired response is fixed analytically: the network solves the active-subframe normal equations.

\subsection{Physical Decoding Time and Task Error}
Let the relaxation dynamics be overdamped,
\begin{equation}
C\dot q=-\nabla V_S(q),\qquad C\succ0.                     \label{eq:overdamped}
\end{equation}

\begin{proposition}[Settling rate]\label{prop:settling}
For fixed $S$, let $e=q-q_S^\star$.  Then
\begin{equation}
\|e(t)\|_C\le e^{-\alpha_St}\|e(0)\|_C,\quad
\alpha_S=\lambda_{\min}(C^{-1/2}J_SC^{-1/2}),             \label{eq:settling}
\end{equation}
where $\|v\|_C^2=v^{\T}Cv$.  Loss of the surviving lower frame bound therefore causes large estimation variance, mechanical softness, and slow decoding simultaneously.
\end{proposition}
\begin{IEEEproof}
Equation \eqref{eq:overdamped} gives $C\dot e=-J_Se$.  With $z=C^{1/2}e$, $\dot z=-C^{-1/2}J_SC^{-1/2}z$.  Spectral decomposition yields \eqref{eq:settling}.
\end{IEEEproof}

\begin{proposition}[Locomotion MMSE]\label{prop:locmmse}
For the posterior-mean decoder and task metric $Q\succeq0$,
\begin{equation}
D_{\loc}(S)\triangleq\E[d_{\loc}(U,q_S^\star)\mid S]
=\tr(QP_S).                                                \label{eq:taskmse}
\end{equation}
For conditionally independent epoch errors in \eqref{eq:destination},
\begin{equation}
\Cov(\widehat D_T-D_T\mid S^T)
=\sum_{t=0}^{T-1}L_tP_{S_t}L_t^{\T}+\Sigma_w.             \label{eq:terminalcov}
\end{equation}
\end{proposition}
\begin{IEEEproof}
Conditioned on the data, $q_S^\star$ is the posterior mean and the error covariance is $P_S$.  The identity $\E[e^{\T}Qe]=\tr(Q\E[ee^{\T}])$ gives \eqref{eq:taskmse}; covariance propagation gives \eqref{eq:terminalcov}.
\end{IEEEproof}

If $a^{\T}U$ is the forward-speed mode and one epoch lasts $\Delta$, then a one-dimensional destination variance is $\Delta^2a^{\T}P_Sa$.  Under a Gaussian approximation,
\begin{equation}
\Pp\{|\widehat D-D|>\epsilon\mid S\}=2Q_{\rm G}\!\left(
\frac{\epsilon}{\Delta\sqrt{a^{\T}P_Sa}}\right),          \label{eq:tail}
\end{equation}
where $Q_{\rm G}$ is the Marcum $Q$-function.  Thus frame conditioning captures the punctual-transport metric used in matter-transport experiments.

\section{Finite-Erasure Gait Design}
For this section, first take an uninformative prior and white residual noise:
\begin{equation}
P_0^{-1}=0,\quad R_S=\sigma^2I.
\label{eq:white}
\end{equation}
For a surviving subframe, define frame bounds
\begin{equation}
A_S=\lambda_{\min}(F_SF_S^{\T}),\quad
B_S=\lambda_{\max}(F_SF_S^{\T}).  \label{eq:framebounds}
\end{equation}
Exact reconstruction requires $A_S>0$; stable reconstruction requires $A_S$ bounded away from zero.

\begin{proposition}[Frame bound to locomotion bound]\label{prop:framebound}
If $A_S>0$, then
\begin{equation}
D_{\loc}(S)=\sigma^2\tr\!\left[Q(F_SF_S^{\T})^{-1}\right]
\le \frac{\sigma^2\tr Q}{A_S}.                           \label{eq:frameerror}
\end{equation}
For isotropic damping $C=cI$, the decoding time constant also satisfies $\tau_S\le c/A_S$.
\end{proposition}

The formula is the Goyal-Kova\v{c}evi\'c-Kelner reconstruction-error criterion expressed in locomotion coordinates; equal-norm tight-frame and optimal-frame refinements were developed further in \cite{GoyalKovacevicKelner2001,CasazzaKovacevic2003,HolmesPaulsen2004}.  A poor lower frame bound means a body mode that is simultaneously weakly constrained, noise amplifying, and slow to settle.

\subsection{One Missing Contact}
\begin{theorem}[One-erasure-optimal gait frames]\label{thm:oneerasure}
Let $F\in\R^{d\times N}$ satisfy $FF^{\T}=I_d$ with $N>d$.  After erasing contact $i$,
\begin{equation}
F_{-i}F_{-i}^{\T}=I_d-f_if_i^{\T},\quad
A_{-i}=1-\|f_i\|^2.               \label{eq:oneerase}
\end{equation}
Consequently,
\begin{equation}
\max_{F:FF^{\T}=I_d}\ \min_i A_{-i}=1-\tfrac{d}{N},       \label{eq:oneoptimal}
\end{equation}
with equality if and only if $F$ is equal norm.
\end{theorem}
\begin{IEEEproof}
The nonzero eigenvalue of $f_if_i^{\T}$ is $\|f_i\|^2$, giving \eqref{eq:oneerase}.  Since $\sum_i\|f_i\|^2=\tr(FF^{\T})=d$, $\max_i\|f_i\|^2\ge d/N$.  Hence $\min_iA_{-i}=1-\max_i\|f_i\|^2\le1-d/N$, with equality exactly when all norms equal $d/N$.
\end{IEEEproof}

The legged locomotion interpretation is direct: under fixed total nominal effort, every leg should be equally dispensable.  A gait in which one contact carries disproportionate energy in a body mode is fragile even if the total number of legs is large.

\subsection{Two and Multiple Missing Contacts}
Define the coherence
\begin{equation}
\mu(F)=\max_{i\ne j}|f_i^{\T}f_j|.                         \label{eq:coherence}
\end{equation}

\begin{theorem}[Two-erasure robustness]\label{thm:twoerasure}
Let $F$ be an equal-norm Parseval frame and set $a=d/N$.  If contacts $i$ and $j$ are erased, then
\begin{equation}
A_{-\{i,j\}}=1-a-|f_i^{\T}f_j|.                           \label{eq:twoerase}
\end{equation}
Hence minimax two-erasure design within this class is equivalent to minimizing coherence.  Whenever an equiangular tight frame exists, it is two-erasure optimal, in accord with Grassmannian frame design \cite{StrohmerHeath2003}.
\end{theorem}
\begin{IEEEproof}
The  eigenvalues of the two-column Gram matrix $f_if_i^{\T}+f_jf_j^{\T}$ are $a\pm|f_i^{\T}f_j|$.  Subtracting this matrix from $I_d$ gives \eqref{eq:twoerase}.  The worst pair is the pair attaining $\mu(F)$.
\end{IEEEproof}

\begin{corollary}[Multiple-erasure coherence bound]\label{cor:multierasure}
If an erasure set $E$ has size $r$, then
\begin{equation}
A_{-E}\ge1-a-(r-1)\mu(F).         \label{eq:multi}
\end{equation}
\end{corollary}
\begin{IEEEproof}
The nonzero eigenvalues of $\sum_{i\in E}f_if_i^{\T}$ equal those of its $r\times r$ Gram matrix.  Gershgorin's circle theorem bounds its largest eigenvalue by $a+(r-1)\mu(F)$.
\end{IEEEproof}

Low coherence means that two legs do not encode nearly the same mixture of body modes.  The statement differs from repetition: repeated contacts are maximally coherent, so they can suppress scalar noise while remaining poor at preserving a multidimensional task.

\subsection{A Harmonic Gait Frame}
The myriapod gait construction due to Chong et al.\ uses spatially-shifted leg and body waves \cite{ChongEtAl2022Gait,ChongEtAl2023}.  This suggests a frame whose coordinates are spatial Fourier modes.

\begin{proposition}[Uniform harmonic gait frame]\label{prop:harmonic}
Let $d=2K+1$, $N\ge2K+1$, and $\theta_i=2\pi i/N$, $i=0,\ldots,N-1$.  Define
\begin{equation}
 f_i=\frac{1}{\sqrt N}\bmat{1\\
 \sqrt2\cos\theta_i\\ \sqrt2\sin\theta_i\\ \vdots\\
 \sqrt2\cos K\theta_i\\ \sqrt2\sin K\theta_i}.        \label{eq:harmonic}
\end{equation}
Then $F=[f_0,\ldots,f_{N-1}]$ is an equal-norm Parseval frame:
\begin{equation}
FF^{\T}=I_d,\quad \|f_i\|^2=d/N.  \label{eq:harmonicparseval}
\end{equation}
\end{proposition}
\begin{IEEEproof}
Discrete Fourier orthogonality makes distinct rows of $F$ orthogonal and gives unit row norm when $N>2K$.  Also $\cos^2(k\theta_i)+\sin^2(k\theta_i)=1$, so each column has squared norm $(1+2K)/N=d/N$.
\end{IEEEproof}

Here $U$ specifies amplitudes of global spatial modes, while $C_i=f_i^{\T}U$ is the command at module $i$.  Contact loss deletes a spatial sample.  The morphology reconstructs the global mode coefficients from the remaining samples, rather than averaging identical thrusts.  Fig.~\ref{fig:finite} illustrates the resulting robustness to finite contact erasures.

\begin{figure}
\centering
\includegraphics[width=\columnwidth]{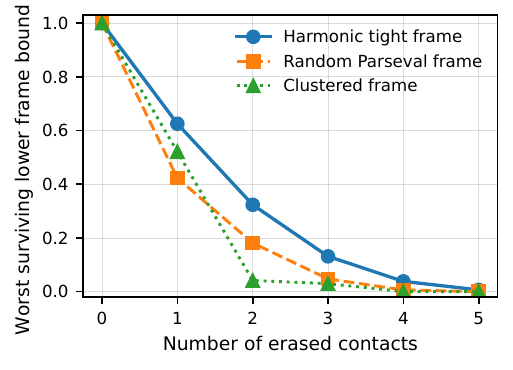}
\caption{Worst surviving lower frame bound for representative $d=3$, $N=8$ Parseval designs.  The equal-norm harmonic construction equalizes single-contact loss and remains comparatively robust to several erasures; clustered columns create redundant contact roles and a weak surviving mode.}
\label{fig:finite}
\end{figure}


\section{A Random-Frame Coding Theorem}
We now let the useful locomotion dimension grow with the number of attempted contacts.  This distinguishes coded locomotion  from scalar repetition \cite{ChongEtAl2023} and connects with probabilistic erasure-robust frame analysis \cite{Vershynin2005}.

For each $N$, let $d_N/N\to R\in(0,1)$ and draw
\begin{equation}
F_N=\frac1{\sqrt N}G_N,\quad (G_N)_{ki}\stackrel{\rm iid}{\sim}\cN(0,1). \label{eq:gaussframe}
\end{equation}
Let contacts survive independently:
\begin{equation}
B_i\stackrel{\rm iid}{\sim}\operatorname{Bernoulli}(q),\quad
M_N=|S_N|.                        \label{eq:bernoulli}
\end{equation}
Define frame-decoding failure in the noiseless model as noninjectivity of $u\mapsto F_{S_N}^{\T}u$.

\begin{theorem}[Random-frame achievability and converse]\label{thm:randomcoding}
Under \eqref{eq:gaussframe}--\eqref{eq:bernoulli}:
\begin{enumerate}[label=(\roman*),leftmargin=4mm]
\item Conditional on $M_N\ge d_N$, $F_{S_N}$ has rank $d_N$ almost surely.  Hence
\begin{equation}
P_{\fail}^{(N)}=\Pp\{M_N<d_N\}.                           \label{eq:pfail}
\end{equation}
\item If $R<q$, then $P_{\fail}^{(N)}\to0$ and
\begin{equation}
-\lim_{N\to\infty}\frac1N\log P_{\fail}^{(N)}
=\Dkl(R\|q),                                              \label{eq:exponent}
\end{equation}
where
\begin{equation}
\Dkl(R\|q)=R\log\tfrac{R}{q}+(1-R)\log\tfrac{1-R}{1-q}.   \label{eq:binaryD}
\end{equation}
\item If $R>q$, then $P_{\fail}^{(N)}\to1$.  More generally, for every deterministic sequence of $d_N\times N$ frame matrices and every decoder, arbitrary commands cannot be recovered whenever $M_N<d_N$.
\end{enumerate}
Thus the exact-recovery analog dimension threshold is
\begin{equation}
C_{\rm frame}=q\quad\text{locomotion dimensions per attempted contact}. \label{eq:capacity}
\end{equation}
\end{theorem}
\begin{IEEEproof}
A Gaussian $d_N\times M_N$ matrix has full row rank almost surely when $M_N\ge d_N$.  This proves \eqref{eq:pfail}.  For $R<q$, the binomial lower-tail large-deviation principle gives \eqref{eq:exponent} \cite{DemboZeitouni1998}.  For $R>q$, the law of large numbers gives $M_N/N\to q<R$, so $M_N<d_N$ with probability tending to one.  For an arbitrary frame, $M_N<d_N$ implies $\rank(F_{S_N})<d_N$, so there is a nonzero $v\in\nullsp(F_{S_N}^{\T})$.  The commands $u$ and $u+v$ produce identical surviving coefficients; no decoder can recover both.
\end{IEEEproof}

\begin{corollary}[Finite relative redundancy]\label{cor:relative}
Every rate $R<q$ is reliably achievable with asymptotic relative redundancy
\begin{equation}
\eta=\frac{N}{d}>\frac1q.                                 \label{eq:finiteredundancy}
\end{equation}
The absolute number of contacts still grows to drive error probability to zero, but contacts per useful locomotion dimension remain finite.
\end{corollary}

For every fixed $N$, the Gaussian frame in \eqref{eq:gaussframe} is full spark almost surely, so it may be drawn once and then fixed as a deterministic gait.  In the erasure-only exact-recovery model, every full-spark frame has failure probability $\Pp\{M_N<d_N\}$; randomness is therefore not essential to the rank threshold or its binomial exponent.  As in other random coding arguments, its role is as an existence and expurgation device and, under residual noise, as a source of surviving subframes with an analyzable singular-value spectrum and a nonzero stiffness margin below $R=q$.  A mechanically feasible random ensemble may likewise be drawn in a basis of admissible body and leg waves and expurgated for poor conditioning.

This is a fundamental limit statement, showing positive useful rate.  A fixed absolute number of legs cannot generally make an independent all-contact-loss event arbitrarily unlikely.  The finite-length sharpening of this threshold is illustrated in Fig.~\ref{fig:phase}.

\subsection{Stable Recovery with Residual Contact Noise}
Let
\begin{equation}
Y_S=F_S^{\T}U+Z_S,\qquad Z_S\sim\cN(0,\sigma^2I).          \label{eq:noisyrandom}
\end{equation}

\begin{theorem}[Random-frame stiffness and noise amplification]\label{thm:randomstable}
Suppose $R<q$.  Almost surely,
\begin{align}
\lambda_{\min}(F_SF_S^{\T})&\longrightarrow(\sqrt q-\sqrt R)^2, \label{eq:lmin}\\
\lambda_{\max}(F_SF_S^{\T})&\longrightarrow(\sqrt q+\sqrt R)^2. \label{eq:lmax}
\end{align}
For the unregularized least-squares decoder on the full-rank event,
\begin{equation}
\lim_{N\to\infty}\frac1{d_N}\E\!\left[\|\widehat U-U\|^2\mid M_N\ge d_N+2\right]
=\frac{\sigma^2}{q-R}.                                   \label{eq:randommse}
\end{equation}
Thus $R\uparrow q$ produces both divergent MSE and a critically soft mechanical mode.
\end{theorem}
\begin{IEEEproof}
Conditional on $M_N=m$, $NF_SF_S^{\T}$ is a $d_N$-dimensional Wishart matrix with $m$ degrees of freedom.  Since $m/N\to q$ almost surely, the Bai--Yin extreme-eigenvalue limits give \eqref{eq:lmin}--\eqref{eq:lmax} \cite{BaiYin1993,TulinoVerdu2004}.  For $m>d_N+1$,
\begin{equation*}
\E[(F_SF_S^{\T})^{-1}\mid M_N=m]=\frac{N}{m-d_N-1}I_d.
\end{equation*}
Taking normalized trace, conditioning on $M_N\ge d_N+2$, and using $M_N/N\to q$ gives \eqref{eq:randommse}.  If the decoder instead uses a fallback with uniformly bounded conditional
mean-square error per mode whenever $M_N\le d_N+1$, then the same limit
holds without conditioning.
\end{IEEEproof}

\begin{figure*}
\centering
\subfloat[Finite-length noninjectivity probability for contact survival $q=0.7$.]{\includegraphics[width=0.47\textwidth]{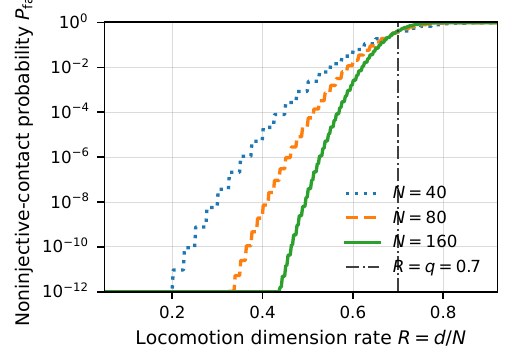}\label{fig:phase}}
\hfill
\subfloat[Below threshold, the MSE diverges and the mechanical stiffness margin vanishes as $R\uparrow q$.]{\includegraphics[width=0.47\textwidth]{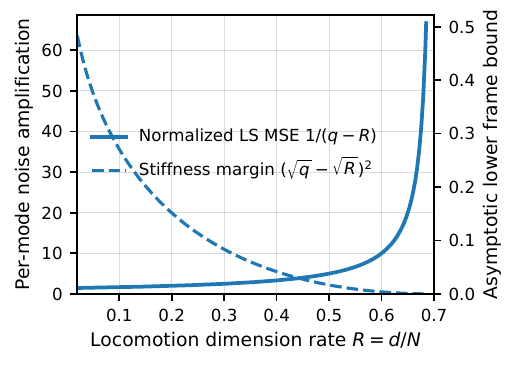}\label{fig:soft}}
\caption{Random-frame coding transition.  Increasing $N$ sharpens the threshold at the surviving-contact fraction.  The same spectral gap that determines numerical decoding stability is the weakest stiffness of the compliant decoder.}
\label{fig:random}
\end{figure*}

The simultaneous noise blow-up and stiffness collapse are shown in Fig.~\ref{fig:soft}.

\begin{proposition}[Bayesian random-frame limit]\label{prop:bayesrandom}
Let $U\sim\cN(0,\sigma_U^2I)$, define $\gamma=\sigma_U^2/\sigma^2$, and let $\nu_{q,R}$ be the limiting eigenvalue distribution of $F_SF_S^{\T}$.  For $R<q$, the asymptotic per-mode MMSE and mutual information are
\begin{align}
D_{\rm mode}(\gamma)&=\sigma_U^2\int\frac{1}{1+\gamma\lambda}\,d\nu_{q,R}(\lambda), \label{eq:mpmmse}\\
\mathcal I_{\rm mode}(\gamma)&=\frac12\int\log(1+\gamma\lambda)\,d\nu_{q,R}(\lambda). \label{eq:mpinfo}
\end{align}
The information per attempted contact is $R\mathcal I_{\rm mode}$.
\end{proposition}
\begin{IEEEproof}
Diagonalize $F_SF_S^{\T}$ in the Gaussian posterior covariance and log-determinant formulas, then use convergence of the empirical spectral distribution to the scaled Marchenko--Pastur law \cite{TulinoVerdu2004}.
\end{IEEEproof}

\subsection{A Distortion Converse}
The rank converse concerns exact recovery.  A noisy, distortion-based benchmark follows from Gaussian rate-distortion theory.

\begin{proposition}[Information required for locomotion fidelity]\label{prop:rdconverse}
Let $\mathbf U\sim\cN(0,\sigma_U^2I_m)$ be a stack of locomotion commands, and let $\widehat{\mathbf U}$ be produced by any gait encoder, terrain interaction, and physical or digital decoder.  If
\begin{equation}
\frac1m\E\|\mathbf U-\widehat{\mathbf U}\|^2\le D<\sigma_U^2, \label{eq:distconstraint}
\end{equation}
then
\begin{equation}
I(\mathbf U;\mathbf Y,S)\ge \frac{m}{2}\log\frac{\sigma_U^2}{D}. \label{eq:rdconverse}
\end{equation}
\end{proposition}
\begin{IEEEproof}
Data processing gives $I(\mathbf U;\mathbf Y,S)\ge I(\mathbf U;\widehat{\mathbf U})$.  The Gaussian rate-distortion converse \cite{Berger1971} lower-bounds the latter by the right side of \eqref{eq:rdconverse}.
\end{IEEEproof}
Comparing to this rate-distortion bound quantifies the price paid by mechanically realizable (linear) frame encoders and quadratic passive decoders relative to an unconstrained source--channel architecture for legged locomotion.

\section{Information, MMSE, and Mechanical Compliance}
The frame operator controls more than reconstruction.  In the Gaussian model it is also an information matrix, so it yields an intriguing direct correspondence among information, estimation, mechanics, and locomotion.  Here all information quantities are evaluated for known realized frames.

\begin{theorem}[Information--compliance identity]\label{thm:infocompliance}
Under \eqref{eq:channel} and \eqref{eq:prior},
\begin{align}
I(U;Y_S\mid S)
&=\frac12\log\frac{\det P_0}{\det P_S}                    \label{eq:infoidentity}\\
&=\frac12\log\frac{\det J_S}{\det P_0^{-1}}.             \label{eq:stiffnessinfo}
\end{align}
If independent contact $i$ has continuously variable precision $\rho_i$, then
\begin{equation}
\frac{\partial}{\partial\rho_i}I(U;Y_S\mid S)
=\frac12 f_i^{\T}P_Sf_i.          \label{eq:precisionderivative}
\end{equation}
\end{theorem}
\begin{IEEEproof}
The mutual information of a linear Gaussian channel is
\begin{equation*}
\frac12\log\det\!\left(I+P_0^{1/2}F_SR_S^{-1}F_S^{\T}P_0^{1/2}\right).
\end{equation*}
The determinant lemma and $P_S=J_S^{-1}$ yield \eqref{eq:infoidentity}--\eqref{eq:stiffnessinfo}.  Differentiating $\frac12\log\det J_S$ gives $\frac12\tr(P_Sf_if_i^{\T})$.
\end{IEEEproof}

Information is therefore the logarithmic shrinkage of the compliance ellipsoid, or equivalently the logarithmic increase in stiffness volume.  A firm contact is informative when it constrains a direction in which the body is currently soft.

\begin{proposition}[Vector I--MMSE relation]\label{prop:immse}
Let
\begin{equation}
Y_S(\gamma)=\sqrt\gamma\,R_S^{-1/2}F_S^{\T}U+N,
\quad N\sim\cN(0,I).              \label{eq:snrchannel}
\end{equation}
For arbitrary finite-power $U$, let
\begin{equation*}
E_S(\gamma)=\E[(U-\E[U\mid Y_S])(U-\E[U\mid Y_S])^{\T}].
\end{equation*}
Then
\begin{equation}
\frac{d}{d\gamma}I(U;Y_S(\gamma)\mid S)
=\frac12\tr\!\left(R_S^{-1/2}F_S^{\T}E_S(\gamma)F_SR_S^{-1/2}\right). \label{eq:immse}
\end{equation}
\end{proposition}
This vector I--MMSE identity of Guo, Shamai, and Verd\'u, with matrix gradients developed by Palomar and Verd\'u \cite{GuoShamaiVerdu2005,PalomarVerdu2006}, shows that the marginal value of improving overall contact quality is half the unresolved error as projected into the active contact coordinates.

\begin{theorem}[Information-to-locomotion inequality]\label{thm:infoloc}
Let $Q\succ0$ and define $D_Q(S)=\tr(QP_S)$.  Then
\begin{equation}
\frac{D_Q(S)}{d}
\ge (\det Q)^{1/d}(\det P_0)^{1/d}
\exp\!\left[-\frac{2}{d}I(U;Y_S\mid S)\right].            \label{eq:infoloc}
\end{equation}
Equality holds if and only if
\begin{equation}
Q^{1/2}P_SQ^{1/2}=cI_d            \label{eq:tasktight}
\end{equation}
for some $c>0$.
\end{theorem}
\begin{IEEEproof}
Apply the arithmetic--geometric mean inequality to the eigenvalues of $Q^{1/2}P_SQ^{1/2}$:
\begin{equation*}
\tfrac1d\tr(QP_S)\ge[\det(QP_S)]^{1/d}.
\end{equation*}
Equation \eqref{eq:infoidentity} gives $\det P_S=\det P_0\,e^{-2I}$.  Equality in arithmetic--geometric mean is exactly \eqref{eq:tasktight}.
\end{IEEEproof}

The theorem distinguishes three useful frame objectives, respectively to minimize average task MSE, maximize contact information, and protect the weakest mechanical mode:
\begin{align*}
\text{A-optimality:}&\quad \min\ \tr(QP_S),\\
\text{D-optimality:}&\quad \max\ \log\det J_S,\\
\text{E-optimality:}&\quad \max\ \lambda_{\min}(J_S).
\end{align*}
For an isotropic task metric and a feasible fixed-trace precision budget, an isotropic precision matrix simultaneously optimizes these criteria. For anisotropic tasks, their optimizers generally differ.

\section{Incremental Redundancy and Morphological HARQ}
In locomotion, one can consider not just feedforward but also feedback-based control analogous to hybrid automatic repeat request (HARQ).  Recent work describes passive force redistribution as mechanical intelligence and combines it with sensor-based control in a forward error correction (FEC)/ARQ analogy \cite{ChongEtAl2025Robust,HeEtAl2025}.  Open-loop frame coding is the forward-error-control component.  Contact sensing or a compliance threshold can add a reverse channel.  The robot may then transmit a new mechanical ``parity symbol'' by executing an additional gait primitive.

After data $\Dset_S$, suppose
\begin{equation}
U\mid\Dset_S\sim\cN(\widehat U_S,P_S).                   \label{eq:currentposterior}
\end{equation}
Candidate primitive $i$ gives
\begin{equation}
Y_i=f_i^{\T}U+Z_i,\quad Z_i\sim\cN(0,\rho_i^{-1}).       \label{eq:candidate}
\end{equation}

\begin{theorem}[One-step incremental redundancy]\label{thm:harq}
Adding candidate $i$ gives
\begin{align}
P_{S+i}&=P_S-
\frac{\rho_iP_Sf_if_i^{\T}P_S}{1+\rho_i f_i^{\T}P_Sf_i}, \label{eq:sherman}\\
\Delta I_i&=I(U;Y_i\mid\Dset_S)
=\frac12\log\!\left(1+\rho_i f_i^{\T}P_Sf_i\right).      \label{eq:gain}
\end{align}
Moreover,
\begin{equation}
D_Q(S)-D_Q(S+i)=
\frac{\rho_i f_i^{\T}P_SQP_Sf_i}
{1+\rho_i f_i^{\T}P_Sf_i}.                                \label{eq:msegain}
\end{equation}
Therefore, among equal-cost candidates, the exact greedy information rule is
\begin{equation}
 i_I^\star=\argmax_i\ \rho_i f_i^{\T}P_Sf_i,              \label{eq:infogreedy}
\end{equation}
and the exact greedy task-MMSE rule is
\begin{equation}
 i_Q^\star=\argmax_i\ 
\frac{\rho_i f_i^{\T}P_SQP_Sf_i}{1+\rho_i f_i^{\T}P_Sf_i}. \label{eq:msegreedy}
\end{equation}
\end{theorem}
\begin{IEEEproof}
Equation \eqref{eq:sherman} is the Sherman--Morrison formula for $(P_S^{-1}+\rho_i f_if_i^{\T})^{-1}$ \cite{sherman1950adjustment}.  The matrix determinant lemma gives \eqref{eq:gain}; tracing $Q(P_S-P_{S+i})$ gives \eqref{eq:msegain}.
\end{IEEEproof}

The score $f_i^{\T}P_Sf_i$ is the mechanical compliance in the direction that primitive $i$ can constrain.  Precision $\rho_i$ discounts a contact that is geometrically valuable but unlikely to be firm.

\begin{corollary}[Weakest-mode targeting]\label{cor:weakest}
Suppose every unit vector $f$ is feasible and every candidate has the same precision.  Then the information-optimal next primitive is
\begin{equation}
f^\star=v_{\max}(P_S)=v_{\min}(J_S),                      \label{eq:weakest}
\end{equation}
namely the mechanically softest unresolved mode.  For $Q=I$, the same direction maximizes one-step MSE reduction.
\end{corollary}
\begin{IEEEproof}
The information score is the Rayleigh quotient $f^{\T}P_Sf$.  For $Q=I$, diagonalizing $P_S$ reduces \eqref{eq:msegreedy} to a linear-fractional function of convex weights on eigenvalues, maximized at the largest eigenvalue because $\rho\lambda^2/(1+\rho\lambda)$ is increasing in $\lambda$.
\end{IEEEproof}

\begin{figure}
\centering
\includegraphics[width=\columnwidth]{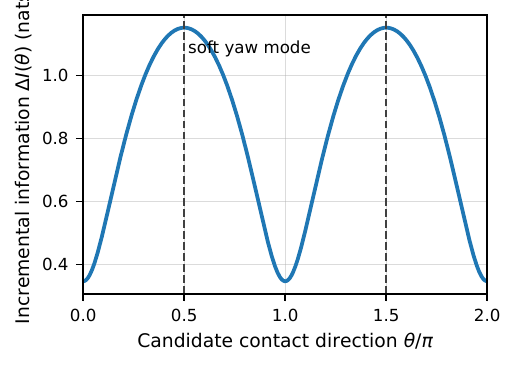}
\caption{Incremental information for a two-mode posterior $P_S=\diag(1,9)$ with unit precision.  A contact direction aligned with the soft second mode supplies the largest conditional information and the largest reduction in uncertainty volume.}
\label{fig:adaptive}
\end{figure}

Fig.~\ref{fig:adaptive} illustrates this weakest-mode targeting rule in two dimensions.
If candidate $i$ succeeds only with probability $q_i$ and costs $c_i$, a natural myopic efficiency score is
\begin{equation}
\frac{q_i}{2c_i}\log\!\left(1+\rho_i f_i^{\T}P_Sf_i\right). \label{eq:costscore}
\end{equation}
This captures the gait--sensor tradeoff in terms of  geometry, contact precision, terrain success probability, time, and energy.

\subsection{Submodularity of Contact Information}
Let a finite library of possible additional primitives be indexed by $\mathcal V$, and define the normalized set function
\begin{equation}
\Phi(A)=\frac12\log\frac{\det\!\left(P_0^{-1}+\sum_{i\in A}\rho_i f_if_i^{\T}\right)}{\det P_0^{-1}}. \label{eq:setinfo}
\end{equation}

\begin{theorem}[Greedy approximation]\label{thm:submod}
The function $\Phi$ is nonnegative, monotone, and submodular.  Under a cardinality budget $|A|\le k$, sequentially choosing the largest marginal gain \eqref{eq:gain} produces a set $A_{\rm gr}$ satisfying
\begin{equation}
\Phi(A_{\rm gr})\ge(1-e^{-1})\max_{|A|\le k}\Phi(A).      \label{eq:submodguarantee}
\end{equation}
\end{theorem}
\begin{IEEEproof}
For $A\subseteq B$, $J_B\succeq J_A$ and hence $J_B^{-1}\preceq J_A^{-1}$.  The marginal gain of adding $i$ is
\begin{equation*}
\frac12\log(1+\rho_i f_i^{\T}J_A^{-1}f_i),
\end{equation*}
which decreases as the selected set grows.  This proves diminishing returns.  The greedy bound is a classical result  \cite{NemhauserWolseyFisher1978}.
\end{IEEEproof}

\subsection{A Zero-Error Variable-Length Theorem}
A rateless gait provides the strongest finite-redundancy statement in the erasure-only model.  At each attempt, generate a new row direction from an absolutely continuous distribution and observe whether the corresponding contact survives.

\begin{theorem}[Rateless gait optimality]\label{thm:rateless}
Let contact attempts survive independently with probability $q$.  Continue until $d$ surviving random frame coefficients have been collected, and let $\tau$ be the number of attempts.  Then reconstruction is exact with probability one and
\begin{equation}
\tau\sim\operatorname{NegBin}(d,q),\quad
\E\tau=\tfrac{d}{q},\quad
\Var(\tau)=\tfrac{d(1-q)}{q^2}.    \label{eq:negativebinomial}
\end{equation}
Moreover, every zero-error variable-length scalar frame scheme with finite expected stopping time satisfies
\begin{equation}
\E\tau\ge\tfrac{d}{q}.             \label{eq:waldconverse}
\end{equation}
Thus the optimal expected redundancy factor is exactly $1/q$.
\end{theorem}
\begin{IEEEproof}
Surviving absolutely continuous random vectors are linearly independent until $d$ have arrived, almost surely.  The stopping time for $d$ Bernoulli successes is negative binomial, giving \eqref{eq:negativebinomial}.  Any zero-error scheme for arbitrary $u\in\R^d$ needs at least $d$ surviving scalar linear observations.  Let $K_\tau=\sum_{i=1}^{\tau}B_i$.  Then $K_\tau\ge d$ almost surely.  Wald's identity gives $\E K_\tau=q\E\tau$, hence \eqref{eq:waldconverse}.
\end{IEEEproof}

Under a hard deadline $n$, the residual failure probability is exactly $\Pp\{\operatorname{Binomial}(n,q)<d\}$, recovering Theorem~\ref{thm:randomcoding}.  Repeating the same $f_i$ is Chase combining from communication theory \cite{chase1985code}, but it is preferable to supply a new complementary $f_i$ as in incremental redundancy, especially when a different body mode remains soft.

\subsection{A Minimal Robophysical Example}
\label{sec:plate-example}

Consider a rigid plate with $N\ge4$ vertical series-elastic legs attached
at equally spaced points on a circle of radius $a>0$.
Write $\xi=(h,\theta,\phi)^T$ for its small heave, pitch, and roll
perturbations about a level, gravity-balanced operating point. Nominal weight
and preload forces are balanced before applying the incremental model,
and the tested active set is held fixed during relaxation.
With attachment azimuths $\alpha_i=2\pi i/N$, $i=0,\ldots,N-1$,
the attachment positions are $(a\cos\alpha_i,a\sin\alpha_i,0)$,
and the affine incremental extension model is
\begin{align}
  \delta\ell_i&=h-a\theta\cos\alpha_i+a\phi\sin\alpha_i=g_i^T\xi,\notag\\
  g_i&=(1,-a\cos\alpha_i,a\sin\alpha_i)^T.
  \label{eq:plate-geometry}
\end{align}
Here the signs fix the pitch and roll conventions; the physical
Jacobians $g_i$ are not themselves the normalized harmonic columns.

Choose a reference displacement $\ell_*>0$ and introduce dimensionless
modal coordinates
\begin{align}
  q&=T\xi,\qquad U=T\xi_{\mathrm{cmd}},\notag\\
  T&=\frac{\sqrt N}{\ell_*}
       \operatorname{diag}\!\left(1,-\frac{a}{\sqrt2},
                                      \frac{a}{\sqrt2}\right),
  \label{eq:plate-coordinates}
\end{align}
where $\xi_{\mathrm{cmd}}$ is the desired physical pose perturbation.
For the $K=1$ harmonic frame,
\begin{equation}
  f_i=\frac1{\sqrt N}
       \begin{bmatrix}1&\sqrt2\cos\alpha_i&\sqrt2\sin\alpha_i\end{bmatrix}^T,
  \qquad FF^T=I_3,
  \label{eq:plate-frame}
\end{equation}
substitution gives the exact identity within the linearized model
\begin{equation}
  \delta\ell_i=\ell_*f_i^Tq.
  \label{eq:plate-length}
\end{equation}
Thus the physical commanded rest-length \emph{increment} is
$\ell_*f_i^TU=g_i^T\xi_{\mathrm{cmd}}$, not simply $f_i^TU$. The same transformation $T$ is retained for all active sets.

Let the effective normalized rest-length data be
$Y_i=f_i^TU+Z_i$, with independent Gaussian errors of variance
$\rho_i^{-1}$, independent of the command and the prescribed contact
selection. For realized data $y_i$, an engaged spring of physical
stiffness $k_i$ has incremental energy
\begin{equation}
  E_i(q)=\frac{k_i\ell_*^2}{2}(f_i^Tq-y_i)^2.
  \label{eq:plate-spring}
\end{equation}
To realize the statistical weights mechanically, select or calibrate
stiffnesses so that
\begin{equation}
  k_i\ell_*^2=\kappa\rho_i
  \label{eq:plate-calibration}
\end{equation}
for one common energy scale $\kappa>0$.
An intrinsic centering mechanism
with energy
$E_0(q)=\kappa(q-\mu)^TP_0^{-1}(q-\mu)/2$
then makes the total incremental energy exactly $E_S(q)=\kappa V_S(q)$.
For $U\sim\mathcal N(\mu,P_0)$ with $P_0\succ0$, mechanical
equilibrium in this model therefore gives
\begin{equation}
  \xi_S^*=T^{-1}q_S^*
         =T^{-1}\mathbb E[U\mid Y_S=y_S,S].
  \label{eq:plate-equilibrium}
\end{equation}
Omitting the centering stiffness instead gives weighted least squares
whenever the surviving columns span $\mathbb R^3$.

Measured stiffness and covariance must also be compared in common
coordinates. With $K_S^{\mathrm{phys}}$ denoting the Hessian of physical
energy with respect to $\xi$, and $\Sigma_{\xi,S}$ the posterior
covariance of the commanded physical pose,
\begin{align}
  K_S^{\mathrm{phys}}&=\kappa T^TJ_ST,\notag\\
  \Sigma_{\xi,S}&=T^{-1}P_ST^{-T}
               =\kappa(K_S^{\mathrm{phys}})^{-1}.
  \label{eq:plate-stiffness}
\end{align}
Likewise, a physical pose-error metric $W_\xi\succeq0$ becomes
$Q=T^{-T}W_\xi T^{-1}$ in the normalized coordinates.
Among equal-cost candidates that
engage successfully, adding the feasible support with largest score
$\rho_i f_i^TP_Sf_i$, while retaining the existing constraints and the
same command, implements the information-greedy update.

\section{Filter-Bank Gaits and Colored Terrain}
The finite-frame model treats a gait epoch as a block.  Continuous locomotion requires streams of body-mode commands and contact actions.  Let $u[t]\in\R^m$ and let an $N\times m$ oversampled analysis filter bank be
\begin{equation}
H(z)=\sum_{\ell=0}^{L_h-1}H[\ell]z^{-\ell}.               \label{eq:filterbank}
\end{equation}
Leg stream $i$ receives
\begin{equation}
c_i[t]=\sum_{r=1}^m\sum_{\ell}h_{ir}[\ell]u_r[t-\ell].    \label{eq:gaitstream}
\end{equation}
Spatial phase shifts, body waves, and contact waves provide physically feasible bases for $H(z)$.

For a fixed surviving stream set $S$, stable noiseless reconstruction requires the filter-bank frame condition
\begin{equation}
A_S I_m\preceq H_S(e^{j\omega})^*H_S(e^{j\omega})\preceq B_S I_m \label{eq:fbframe}
\end{equation}
for almost every $\omega$, with $A_S>0$ \cite{KovacevicDragottiGoyal2002,Vaidyanathan1993}.  The canonical synthesis response is
\begin{equation}
G_S(e^{j\omega})=
[H_S^*(e^{j\omega})H_S(e^{j\omega})]^{-1}H_S^*(e^{j\omega}). \label{eq:fbsynthesis}
\end{equation}

Now let the desired gait stream and terrain noise be jointly stationary Gaussian with spectral densities $S_U(\omega)$ and $S_Z(\omega)$.  The posterior error spectrum is
\begin{equation}
P_S(\omega)=\left[S_U^{-1}(\omega)+H_S^*(\omega)S_Z^{-1}(\omega)H_S(\omega)\right]^{-1}, \label{eq:posteriorspectrum}
\end{equation}
and the mutual information rate is
\begin{equation}
\mathcal I_S=\frac1{4\pi}\int_{-\pi}^{\pi}
\log\det\!\left[I+S_U^{1/2}H_S^*S_Z^{-1}H_SS_U^{1/2}\right]d\omega. \label{eq:inforate}
\end{equation}
A task transfer matrix $L(\omega)$ gives locomotion distortion rate
\begin{equation}
\mathcal D_{\loc}=\frac1{2\pi}\int_{-\pi}^{\pi}
\tr[L(\omega)P_S(\omega)L^*(\omega)]d\omega.             \label{eq:distortionrate}
\end{equation}

This formulation clarifies the moving-average observation.  Equal spatial averaging suppresses high-spatial-frequency noise but leaves common-mode or low-frequency terrain disturbances.  Whitening instead weights modes by $S_Z^{-1}(\omega)$.  Mechanically, diagonal precision corresponds to independent leg springs; off-diagonal or frequency-selective precision requires cross-leg elastic or viscoelastic coupling.  Interleaving the coefficients of one body mode across spatially and temporally separated contacts can transform bursty terrain loss into dispersed erasures.


Just as a note, an arbitrary Wiener synthesis bank is not automatically realizable by a passive, causal morphology.  Positive-realness, locality, actuator limits, unilateral contact, and network-synthesis constraints must be imposed.  Pushing beyond the quasi-static rank-one construction of Lemma~\ref{lem:rankone} requires a separate problem of exact dynamic mechanical realization.

\section{Conclusion}
A multilegged robot need not use every leg to repeat the same thrust action.  It can distribute a multidimensional locomotion command across heterogeneous contacts as a finite-frame expansion.  Rough terrain then causes coefficient erasures and noise, while a contact-gated compliant body physically synthesizes the command from the surviving subframe as a form of embodied intelligence.  In the linear-Gaussian setting, equilibrium is the MMSE estimate, stiffness is posterior precision, compliance is posterior covariance, and mutual information is logarithmic stiffness-volume gain.

This formulation yields both fundamental limits and optimal designs.  Equal-norm tight frames make contacts equally dispensable; low coherence protects against multiple losses; and harmonic frames give a gait-compatible construction.  Random frames provide reliable recovery below the surviving-contact fraction and a rank converse above it, establishing finite redundancy per useful locomotion dimension.  Near the threshold, noise amplification and physical softness diverge together.  Feedback can then add only the contact constraint that addresses the weakest unresolved task mode, with exact one-step and greedy approximation guarantees.

The present results assume a linearized command-to-contact map, quasi-static equilibrium, and a task-relevant body coordinate system of fixed dimension within each block.  They also assume contact loss gates an elastic path cleanly, and that slipping or partial contacts are a form of Gaussian noise.  Future work should consider properties of real robots that add friction cones, saturation, collision constraints, inertia, stability failures, and contact-dependent nonlinear geometry.  These effects may change the feasible frame class and may make the energy nonquadratic.  In that case equilibrium is generally a MAP-like robust estimator rather than exact MMSE.

\bibliographystyle{IEEEtran}
\bibliography{references}

\appendix[Further Proof Details]
\subsection{Proof of Proposition~\ref{prop:framebound}}
Under \eqref{eq:white}, the least-squares error is
\begin{equation*}
\widehat U-U=(F_SF_S^{\T})^{-1}F_SZ_S.
\end{equation*}
Its covariance is $\sigma^2(F_SF_S^{\T})^{-1}$.  Proposition~\ref{prop:locmmse} gives the equality in \eqref{eq:frameerror}.  Since $(F_SF_S^{\T})^{-1}\preceq A_S^{-1}I$, the upper bound follows.  The settling-time statement is Proposition~\ref{prop:settling} with $J_S=\sigma^{-2}F_SF_S^{\T}$ after absorbing the common stiffness scale into $c$.

\subsection{Finite-Length Bounds for Theorem~\ref{thm:randomcoding}}
For $d_N/N<R_0<q$, Chernoff's inequality gives
\begin{equation}
P_{\fail}^{(N)}
=\Pp\{M_N<d_N\}
\le \exp[-N\Dkl(d_N/N\|q)].                              \label{eq:chernoff}
\end{equation}
A matching logarithmic lower bound follows from the method of types, yielding \eqref{eq:exponent}.  Conversely, if $R>q$, then for any $\delta\in(0,R-q)$,
\begin{equation*}
\Pp\{M_N\ge d_N\}
\le\Pp\{M_N/N\ge q+\delta\}\to0.
\end{equation*}
The strong converse is thus a dimension-counting statement coupled to concentration of the number of surviving contacts.

\subsection{Posterior-Mean Update}
The posterior mean after candidate $i$ can be written in Kalman form:
\begin{equation}
\widehat U_{S+i}=\widehat U_S+
\frac{P_Sf_i}{\rho_i^{-1}+f_i^{\T}P_Sf_i}
\left(Y_i-f_i^{\T}\widehat U_S\right).                   \label{eq:kalmanupdate}
\end{equation}
The term in parentheses is the mechanical innovation.  In a passive realization, the new spring adds $\rho_i f_if_i^{\T}$ to stiffness and relaxation performs \eqref{eq:kalmanupdate} without an explicit matrix inversion.

\end{document}

%% file: figures/fig1.tikz
\begin{tikzpicture}[
  x=1cm,y=1cm,
  font=\fontsize{8.5}{10}\selectfont,
  line cap=round,line join=round,
  >={Latex[length=1.7mm,width=1.2mm]},
  panel/.style={draw=black!25,fill=black!2,rounded corners=2pt,
                line width=.45pt},
  title/.style={font=\sffamily\bfseries\fontsize{8.5}{10}\selectfont,
                align=center,text width=3.9cm},
  note/.style={font=\fontsize{8}{9.5}\selectfont,
               align=center,text width=3.7cm},
  box/.style={draw=black!65,fill=white,rounded corners=2pt,
              line width=.6pt,align=center,inner sep=4pt},
  flow/.style={->,line width=.85pt},
  feedback/.style={->,line width=.7pt,densely dashed},
  spring/.style={line width=.75pt,decorate,
     decoration={coil,aspect=.42,segment length=3pt,amplitude=2.3pt}},
  support/.style={line width=.65pt},
  missing/.style={draw=black!45,densely dashed,line width=.65pt}
]

\foreach \x in {0,4.5,9,13.5}
  \path[panel] (\x,0) rectangle ++(4.1,5.55);

\begin{scope}
  \node[title] at (2.05,5.20) {(a) FRAME ENCODING};
  \node[note] at (2.05,4.83) {body modes to local commands};

  \node[box,minimum width=2.95cm,minimum height=.62cm]
    (source) at (2.05,4.17) {$U\in\mathbb R^d$};
  \node[box,minimum width=2.95cm,minimum height=.62cm]
    (analysis) at (2.05,3.10) {$F^T$\,: linear analysis};
  \draw[flow] (source.south) -- (analysis.north);

  \coordinate (bus) at (2.05,2.52);
  \draw[support] (analysis.south) -- (bus);
  \draw[support] (.68,2.52) -- (3.42,2.52);
  \foreach \x/\lab in {.68/1,1.67/2,3.42/N}{
    \node[box,minimum width=.62cm,minimum height=.48cm,
          inner sep=2pt] (coef\lab) at (\x,1.97) {$C_{\lab}$};
    \draw[flow] (\x,2.52) -- (coef\lab.north);
  }
  \node at (2.58,1.97) {$\cdots$};
  \node at (2.05,1.28) {$C_i=f_i^TU$};
  \node[note] at (2.05,.53)
    {heterogeneous mixtures\\
     $R_{\mathrm{loc}}=d/N$};
\end{scope}

\begin{scope}[xshift=4.5cm]
  \node[title] at (2.05,5.20) {(b) TERRAIN CHANNEL};
  \node at (2.05,4.60) {$Y_S=F_S^TU+Z_S$};

  \path[draw=black!70,fill=black!12,line width=.7pt]
    (.35,3.64) rectangle (3.75,3.89);
  \draw[support] (.65,3.64) -- (.48,2.78) -- (.67,1.91);
  \draw[missing] (1.45,3.64) -- (1.64,2.80) -- (1.49,2.02);
  \draw[support] (2.45,3.64) -- (2.20,2.82) -- (2.42,2.02);
  \draw[missing] (3.35,3.64) -- (3.51,2.80) -- (3.29,1.94);

  \draw[line width=.9pt]
    (.22,1.85) -- (.67,1.91) -- (1.02,1.89)
    -- (1.21,1.34) -- (1.83,1.34) -- (2.01,1.94)
    -- (2.42,2.02) -- (2.95,1.86) -- (3.29,1.94)
    -- (3.87,2.00);
  \fill (.67,1.91) circle (1.8pt);
  \fill (2.42,2.02) circle (1.8pt);
  \draw[fill=white,line width=.7pt] (3.29,1.94) circle (2pt);
  \draw[line width=.75pt] (1.39,1.92) -- (1.59,2.12);
  \draw[line width=.75pt] (1.39,2.12) -- (1.59,1.92);

  \node[font=\fontsize{8}{9}\selectfont] at (.66,1.01) {firm};
  \node[font=\fontsize{8}{9}\selectfont] at (1.51,1.01) {erased};
  \node[font=\fontsize{8}{9}\selectfont] at (3.30,1.01) {weak};
  \node[note] at (2.05,.47)
    {$S=\{i:B_i=1\}$\\weak contacts: smaller $\rho_i$};
\end{scope}

\begin{scope}[xshift=9cm]
  \node[title] at (2.05,5.20) {(c) COMPLIANT DECODER};
  \node[note] at (2.05,4.76) {fixed active set $S$};
  \node at (2.05,4.34) {body coordinates $q$};
  \path[draw=black!70,fill=black!12,line width=.7pt]
    (.40,3.86) rectangle (3.70,4.04);
  \foreach \x in {.70,2.37,3.35}{
    \draw[support] (\x,3.86) -- (\x,3.64);
    \draw[spring] (\x,3.64) -- (\x,2.96);
    \draw[support] (\x,2.96) -- (\x,2.77);
    \draw[support] (\x-.18,2.77) -- (\x+.18,2.77);
    \foreach \dx in {-.13,0,.13}
      \draw[line width=.4pt]
        (\x+\dx,2.77) -- ++(-.08,-.10);
  }
  \draw[missing] (1.50,3.86) -- (1.50,3.50);
  \draw[missing] (1.50,3.15) -- (1.50,2.77);
  \draw[line width=.65pt] (1.41,3.24) -- (1.59,3.42);
  \draw[line width=.65pt] (1.41,3.42) -- (1.59,3.24);
  \node[note] at (2.05,2.35) {each active path adds $\rho_i f_i f_i^T$};

  \node at (2.05,1.70)
    {$\displaystyle J_S=P_0^{-1}+\sum_{i\in S}\rho_i f_i f_i^T$};
  \node at (2.05,1.08)
    {$\widehat U=q_S^\star,\qquad P_S=J_S^{-1}$};
  \node[note] at (2.05,.45)
    {stiffness $\leftrightarrow$ precision\\
     compliance $\leftrightarrow$ covariance};
\end{scope}

\begin{scope}[xshift=13.5cm]
  \node[title] at (2.05,5.20) {(d) TASK AND FEEDBACK};
  \node at (2.05,4.60) {$x^+=x+L\widehat U+w$};

  \fill (.45,3.33) circle (1.8pt);
  \node[anchor=north] at (.45,3.18) {$x$};
  \draw[flow] (.45,3.33)
    .. controls (1.25,3.98) and (2.05,3.00) .. (3.17,3.53);
  \draw[line width=.7pt,rotate around={18:(3.24,3.54)}]
    (3.24,3.54) ellipse (.40 and .24);
  \fill (3.24,3.54) circle (1.6pt);
  \node[anchor=south] at (3.24,3.83) {$x^+$};
  \node[note] at (2.05,2.67)
    {$D_{\mathrm{loc}}(S)=\operatorname{tr}(QP_S)$};

  \node[box,minimum width=3.66cm,minimum height=1.15cm,
        text width=3.42cm] at (2.05,1.48)
    {\textbf{Information-greedy choice}\\[4pt]
     $\displaystyle i_I^\star\in\arg\max_i\,
       \rho_i f_i^TP_Sf_i$};
  \node[note] at (2.05,.47)
    {feasible direction,\\
     largest information gain};
\end{scope}

\foreach \x in {4.1,8.6,13.1}
  \draw[flow] (\x+.03,2.65) -- ++(.34,0);

\draw[feedback]
  (15.55,0) -- (15.55,-.70) -- (6.55,-.70) -- (6.55,0);
\node[fill=white,inner sep=3pt,align=center,
      font=\fontsize{8}{9.5}\selectfont,text width=6.75cm]
  at (11.05,-.70)
  {optional: next contact for the same $U$\\
   previous constraints must remain available};
\end{tikzpicture}